\documentclass[authoryear,a4paper,12pt]{elsarticle}
\usepackage{natbib}
\usepackage{amsthm}
\usepackage{amssymb}
\usepackage{graphicx,amsmath}

\usepackage{mathtools}
\usepackage{natbib}

\usepackage{footmisc} 
\usepackage{booktabs,xcolor}
\usepackage{booktabs}
\usepackage{multirow}
\usepackage{graphics}
\usepackage{enumitem}
\usepackage{subcaption}
\usepackage{appendix}
\usepackage{lscape}
\usepackage{adjustbox}
\usepackage[flushleft]{threeparttable}
\usepackage{multirow}
\usepackage{setspace} 
\usepackage{caption}
\usepackage[colorlinks=true,linkcolor=blue, urlcolor=blue, citecolor=blue, breaklinks]{hyperref}
\journal{Games and Economic Behavior}

\usepackage{etoolbox}
\makeatletter
\patchcmd{\ps@pprintTitle}
{Preprint submitted to}
{Revised and resubmitted to}
{}{}
\makeatother
\patchcmd{\emailauthor}{(#2)}{}{}{}
\patchcmd{\urlauthor}{(#2)}{}{}{}

\newtheorem{proposition}{Proposition}

\theoremstyle{definition}

\newtheorem*{lemma*}{Lemma}

\newcommand{\ttc}{\text{TTC}} 
\newcommand{\da}{\text{DA}} 
\newcommand{\rmm}{\text{RM}} 
\newcommand{\rsd}{\text{RSD}} 
 
\newcommand{\rk}{\textup{rk}} 
\newcommand{\EE}{\mathbb{E}}
\newtheorem*{theorem*}{Theorem}
\theoremstyle{plain} 
\newcommand{\thistheoremname}{}
\newtheorem*{genericthm*}{\thistheoremname}
\newenvironment{namedthm*}[1]
{\renewcommand{\thistheoremname}{#1}%
	\begin{genericthm*}}
	{\end{genericthm*}}


\begin{document}
	\begin{frontmatter}
		
		\title{The Cost of Strategy-Proofness in School Choice}

		\author[a1]{Josu\'e Ortega}\ead{j.ortega@qub.ac.uk (corresponding author).}		
		\author[a2,a3]{Thilo Klein}

		\address[a1]{Queen’s University Belfast, UK.}	
		\address[a2]{ZEW -- Leibniz Centre for European Economic Research, Mannheim, Germany.}
		\address[a3]{Pforzheim University, Germany.}

		\date{\today}

		\begin{abstract}
			
		We compare the outcomes of the most prominent strategy-proof and stable algorithm (Deferred Acceptance, DA) and the most prominent strategy-proof and Pareto optimal algorithm (Top Trading Cycles, TTC) to the allocation generated by the rank-minimizing mechanism (RM).
			While one would expect that RM improves upon both DA and TTC in terms of rank efficiency, the size of the improvement is nonetheless surprising. 
			Moreover, while it is not explicitly designed to do so, RM also significantly improves the placement of the worst-off student. Furthermore, RM generates less justified envy than TTC.	
			We corroborate our findings using data on school admissions in Budapest. 
		\end{abstract}
		
	\end{frontmatter}

\setcounter{footnote}{0}

\parskip 12pt
\section{Introduction}

School choice is a common way to assign students to schools based on the students' and schools' preferences. Students and schools rank their potential matches and submit this information to a centralized clearinghouse. Afterwards, an algorithm (also known as a mechanism) is applied to the submitted data and an allocation of students to schools is generated.

But which mechanism should we use to assign students to schools? Several economists have argued that a key criterion is that such mechanism must be strategy-proof, i.e. it should not give incentives to students to misrepresent their preferences. Strategy-proofness is a desirable property because it levels the playing field across sophisticated and unsophisticated applicants, while at the same time makes it possible for education authorities to provide clear advice on how to rank schools. The student-proposing deferred acceptance (DA) and top trading cycles (TTC) algorithms have been proposed and implemented in real-life  largely because both are strategy-proof \citep{abdulkadirouglu2003}. 

At the same time, strategy-proofness is a costly property that is incompatible with a variety of other desiderata, and often leads to efficiency losses.\footnote{\cite{abdulkadirouglu2009strategy} document that around 2\% of students in the New York City high school match could be assigned to a more desirable school with a non strategy-proof mechanism without affecting  the placement of other  students. See \cite{abdulkadiroglu2022school} for a summary of impossibility results on combining strategy-proofness with other desiderata.} In this paper, we aim to quantify the cost of strategy-proof mechanisms in terms of efficiency and equality, both in theory and in practice. To do so, we compare the expected outcomes of the most prominent strategy-proof and stable algorithm (DA) and the most prominent strategy-proof and Pareto optimal algorithm (TTC) to an allocation generated by the rank-minimizing mechanism (RM, \citealp{featherstone2020rank}), which selects an allocation that minimizes the average rank of schools to which students are assigned, without taking schools' priorities into account. We emphasize that RM is neither strategy-proof or stable. On the other hand, our results show that i) RM assigns the average student to a school they prefer more (i.e. it is more efficient)\footnote{Throughout the paper, we write efficiency to refer to rank-efficiency, which is a stronger notion than ordinal efficiency and Pareto efficiency \citep{featherstone2020rank}.}, and ii) RM assigns the worst-off student to a school that they prefer \emph{much} more (i.e. it is \emph{much} more egalitarian). 

In particular, if there are $n$ students and $n$ schools with one seat each, and preferences for both sides are drawn uniformly at random, TTC and DA asymptotically assign the average student to approximately their $\log(n)$ most preferred school, whereas RM  assigns them to a school better than their second choice, independently of the size of $n$. If we focus on the worst placement, rather than the average, the difference is even bigger: RM assigns the worst-off student to his $\log_2(n)$ most preferred school, whereas DA assigns him to his $\log^2(n)$ most preferred school. TTC does even worse, assigning some student to a school in the bottom half of his preference  list (see Fig. \ref{fig:histogram} for the rank distribution). 
\begin{figure}[h!]
	\centering
	\includegraphics[width=\textwidth]{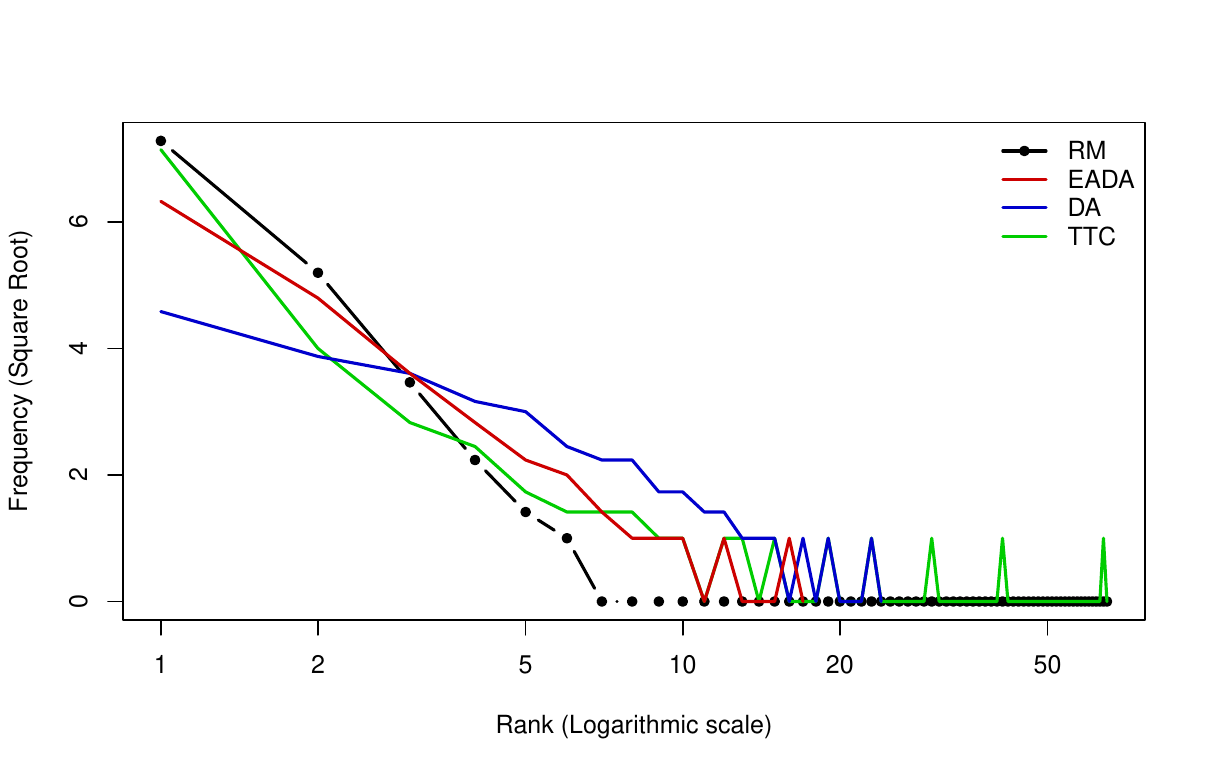} 
	\caption{Rounded average rank distribution in 1,000 iid random markets with $n=100$. \tiny The rank vector obtained in each market is sorted in descending order, and then we compute the average  across all markets coordinate-wise (results are rounded to the nearest integer). The x-axis is truncated at the highest value with positive density. EADA refers to Kesten's efficiency adjusted DA. See Appendix B for details.}
	\label{fig:histogram}
\end{figure}

Our results documenting the inefficiency and inequality in TTC extend to any other Pareto optimal and strategy proof mechanisms in large markets (including random serial dictatorship) because all such mechanisms produce asymptotically the same rank distribution  \citep{pycia2019evaluating}. Similarly, because student-proposing DA produces the best stable matching for students, any other strategy-proof and stable mechanism would generate an allocation with a higher (i.e. worse) average and maximum rank.

RM is Pareto optimal for the students, unlike DA, and generates justified envy for fewer students than TTC, which is surprising because RM does not use schools' priorities but TTC does. We prove these properties for random one-to-one markets where preferences are drawn independently and uniformly at random  (see Table \ref{tab:maintable}), and document them by analyzing real data from the many-to-one student assignment system in Budapest, in which preferences are highly correlated.

Throughout the paper, we assume for simplicity that students report their preferences truthfully in the RM mechanism, which is not strategy-proof. Thus, our results can be interpreted as the cost of strategy-proofness compared to an ideal first-best scenario that may not able attainable. Nonetheless, we document that the rank distribution in RM changes only minimally when large shares of students misrepresent their preferences, and therefore we conjecture that strategic behavior is unlikely to significantly alter the average and maximum ranks in RM.

Moreover, RM is not obviously manipulable \citep{troyan2022non}, meaning that cognitively limited subjects may not be able to find how to manipulate this mechanism successfully. A growing literature has focused on not-obviously manipulable mechanisms, especially the efficiency-adjusted deferred acceptance (EADA) mechanism \citep{kesten2010school}, which is more efficient than DA in theory, and more truthful and efficient than DA in the lab \citep{cerrone2022school}. We compare EADA versus RM, DA, and TTC empirically and in simulations, and find that it improves the efficiency of DA, although not as much as RM, without improving on DA's inequality. Computing the average and maximum rank generated by EADA in random markets remains a challenging open question.\footnote{During the review process of this article, \cite{ortegaziegler} have computed the average rank for EADA. See the note before the acknowledgements for details.}
\begin{table}[h!]
	\centering
	\caption{Theoretical properties of school choice mechanisms in large random markets.}
	\label{tab:maintable}
	\begin{tabular}{lccc}
		\toprule
		& RM & TTC    & DA\\
		\midrule
		Average rank& $\leq 2$ & $\log(n)$&$\log(n)$\\	
		Maximum rank & $\log_2(n)$&$> 0.5 \, n$&$\log^2(n)$\\
		Students w. justified envy &  $0.33 \, n$ &	$0.39 \, n$	 &0\\
Pareto optimal & Yes & Yes & No\\
Not obviously manipulable & Yes & Yes & Yes\\
Strategy-proof & No & Yes & Yes\\

		\bottomrule
	\end{tabular}
\end{table}

\section{Related Literature}

The question of which mechanism should be used to assign students to schools has been frequently asked. The answer to this question in the market design literature is that all frequently used mechanisms generate equivalent rank distributions. This equivalence has been established theoretically for a wide class of mechanisms \citep{che2018payoff, pycia2019evaluating} and empirically using  real-life data \citep{abdulkadirouglu2009strategy, pathak2013school, che2018payoff, aeri}. Our paper challenges the literature consensus by showing that Pareto optimal mechanisms are not equivalent, as can be observed in Fig. \ref{fig:histogram}.\footnote{A two-sample Kolmogorov-Smirnov test rejects the null-hypothesis that any two distributions in Fig. \ref{fig:histogram} are the same at the 1\% significance level.} Three reasons explain the discrepancy between our results and those in the literature, namely i) different model specifications,  ii) we consider a non-strategy proof mechanism such as RM, and iii) RM has not been used in empirical studies. We explain these differences in detail below.

The closest paper to ours is \cite{che2018payoff}. Using a random market approach, they show that the normalized payoff distribution generated by any Pareto optimal mechanism is asymptotically equivalent. Furthermore, they compare the rank distribution generated by DA and TTC (but not RM) using data from the New York City school choice program. The main lesson from their paper is that all Pareto optimal mechanisms are equivalent in large markets, and therefore there is no reason to prefer any Pareto optimal mechanism over another. Our paper shows that the equivalence between Pareto optimal mechanisms breaks down once i) ranks are used instead of normalized payoffs, and ii) students are allowed to rank all available schools, rather than just a few. 

\cite{pycia2019evaluating} obtains a similar equivalence result to that of Che and Tercieux: he shows that any anonymous statistics, such as rank distribution, generated by Pareto efficient and strategy-proof mechanisms are equivalent in large markets.\footnote{Pyicia's result builds on a previous, more general result by \cite{liu2016}.} This implies that all of our results for TTC's poor performance with regards to efficiency and equality also apply to a wide number of mechanisms, including the random serial dictatorship mechanism (RSD),  which ``{\it has a long history and is used in a wide variety of practical allocation problems, including school choice, worker assignment, course allocation, and the allocation of public housing}'' \citep{pycia2021theory}. 

To show that RM is more efficient than DA and TTC, we connect the school choice problem to that of assigning one of $n$ jobs to each of $n$ workers so to minimize costs.\footnote{A large literature in mathematics, uncited in economics, has studied this problem. See \cite{olin1992asymptotic} and \cite{krokhmal2009random} for a summary of it.} Worker $i$ incurs a cost $c_{ij}$ when completing job $j$. The matrix $C$ contains all such costs. When each row of $C$ is an independent random permutation of $\{1, \ldots, n\}$, this problem is equivalent to that of finding the rank-minimizing allocation of students to schools, ignoring schools' priorities. Each entry $c_{ij}$ denotes the rank (cost) of school (job) $j$ for student (worker) $i$. To show that the RM is more efficient than TTC and DA, we invoke a result in \cite{parviainen_2004} which shows that the cost-minimizing allocation has an average cost smaller than 2, and compare it with the well-known average rank in TTC and DA, which is around $\log(n)$. Obtaining the maximum rank lower bound and the fraction of students with justified envy is easy using the limit distribution of ranks in RM, which is also provided by Parviainen.

The result of average rank being bounded in RM was recently independently discovered by \cite{nikzad2022rank}, who provides a bound of 7.75 (ours is 2). His proof uses random graphs and is different (and significantly more involved) than ours. \citet{sethuramannote} shows that Nikzad's bound can be improved to 2 using the cost assignment problem with costs distributed in $(0,1)$ \citep{aldous2001zeta}, without using Parviainen's result.  These papers do not study the maximum rank and justified envy in RM, TTC and DA, and do not analyse the performance of these three mechanisms using real-life data in which preferences are correlated.

The RM mechanism was first studied in economics by \cite{featherstone2020rank}. He documents that RM has been used in practice to assign teachers to schools in the US, and shows that any selection of the RM mechanism cannot be strategy-proof. Nonetheless, he shows that truth-telling is a best response in RM when students have little information about other students' preferences and do not truncate their preference list. He shows that a rank-efficient allocation must be ordinally efficient, but the converse is not necessarily true. He also shows that an inefficient assignment can converge to the RM outcome by performing local swaps. \cite{troyan2022non} has recently shown that RM is not obviously manipulable, meaning that although potential manipulations exist, they cannot be recognized by cognitively limited agents. Therefore, RM has better incentives properties than the well-known Boston mechanism, which is obviously manipulable.

The fact that DA is inefficient is well-known: \cite{kesten2010school} shows that, in a worst-case scenario, it may assign each student to her worst or second-worst school. We show that DA is also inefficient in an average-case scenario. The inefficiency of TTC is less known, partially because the matching literature often focuses on the weaker efficiency notion of Pareto optimality. Nonetheless, \cite{manea2009asymptotic} has shown that the number of preference profiles for which RSD is ordinally efficient (a weaker notion than rank efficiency) vanishes when the number of agents grows. Our result complement his by showing that any Pareto optimal and strategy-proof mechanism not only rarely produces an ordinal efficient allocation (and thus rank-efficient), but also the size of its inefficiency does not vanish in large markets. To our knowledge, the inequality of both mechanisms has remained largely unstudied in the economics literature (with the very recent exception of \cite{galichon2023stable}, who document the inequality of DA for a specific class of preferences).

TTC minimizes justified envy among all Pareto optimal and strategy-proof mechanisms \citep{aeri}. Neither DA nor RM are in this class of mechanisms. We find theoretically that fewer students experience justified envy in RM than in TTC. In practice RM and TTC generate roughly the same amount of justified envy.

\section{Model}
We study a standard one-to-one\footnote{The one-to-one assumption is commonly used for simplicity in the literature (see \cite{roth1992two}) and is not crucial for our results. See the Appendix for a robustness exercise.} school choice market \citep{abdulkadirouglu2003}, which consists of:
\begin{enumerate}
	\item A set of students $T= \{1,\ldots, n\}$,
	\item A set of schools $S= \{s_1,\ldots, s_n\}$, with each school having space for one student only,
	\item Strict students' preferences over schools $\succ \coloneqq (\succ_1,\ldots, \succ_n)$, and
	\item Strict schools' priorities over students $\triangleright \coloneqq (\triangleright_{s_1}, \ldots, \triangleright_{s_n})$.
\end{enumerate}

An allocation $x$ is a perfect matching between $T$ and $S$. We will denote by $x_t$ the school to which student $t$ is assigned, and by $x_s$ the student that school $s$ is assigned to. Student $t$ experiences justified envy in allocation $x$ if there exists a school $s$ such that $s \succ_t x_t$ and $t \triangleright_{s} x_s$.

The function $\rk_t(x_t)$ returns an integer between 1 and $n$ corresponding to the ranking of $x_t$ in the preference list of student $t$, i.e. the most desirable option gets a ranking of 1, whereas the least desirable one gets a ranking of $n$.\footnote{There is a large literature that uses rank distributions as a welfare measure, e.g. \cite{knoblauch2009marriage, ashlagi2017, ortega2018, ortega2019}.} A mechanism is a map from school choice markets to (a probability distribution over) allocations. An allocation $x$ Pareto dominates a different allocation $y$ if, for every student $t$, $\rk_t(x_t) \leq \rk_t (y_t)$ and for some student $j$, $\rk_j(x_j) < \rk_j (y_j)$. An allocation is Pareto optimal if it is not Pareto dominated. A Pareto optimal mechanism returns a Pareto optimal allocation in every school choice problem.

We use $x^\rmm$ to denote one of the (possibly many) allocations that minimizes the sum of ranks for students, which we henceforth call rank-efficient or rank-minimizing. $X^\rmm$ denotes the set of all rank-efficient allocations. The rank-minimizing mechanism (henceforth RM) returns a rank-efficient allocation for every matching market. We assume that RM implements one among all rank-efficient allocations randomly.\footnote{Different rank-efficient allocations may have distinct maximum ranks and number of agents with justified envy. Our results concern the expected properties of rank-efficient allocations, rather than of specific realizations (although we find that the rank distribution changes minimally across efficient allocations; the variance in RM is conjectured to be very small, in the order of $1/n$ \citep{parviainen_2004}).}

Two other mechanisms are of interest. The first is top trading cycles (TTC), in which the following two steps are repeated until all agents have been assigned an object:
\begin{enumerate}
	\item Construct a graph with one vertex per student or school.  Each student (resp. school) points to their top-ranked school (resp. student) among the remaining ones. At least one cycle must exist and no two cycles overlap. Select the cycles in this graph.
	
	\item Permanently assign each student in a cycle to the school they point to. Remove all students and schools involved in a cycle.
\end{enumerate}

The second mechanism of interest is student-proposing deferred acceptance (DA). It works as follows:
\begin{enumerate}
	\item All unmatched students apply to their most preferred school that has not rejected them. Each school that has received a proposal puts the one sent by the highest priority student in a waiting list and permanently rejects all other received applications (if any). 
	\item Repeat step 1 until all schools have received at least one application. Assign each student to the school which has them on a waiting list.
\end{enumerate}

We use $x^\ttc$ and $x^\da$ to denote the allocation obtained by the TTC and DA mechanisms, respectively. Schools' priorities are used to compute TTC and DA, but are irrelevant in RM.

In this paper, we will focus on comparing DA and TTC to RM. However, we will also provide empirical results for another well-studied mechanism called efficiency-adjusted deferred acceptance (EADA, \citealp{kesten2010school}), which like RM, is not strategy-proof. In DA, a student $t$ is called an interrupter if he applies to a school $s$, causing another student to be rejected from $s$, but eventually being rejected himself from $s$ at a later round of DA. The efficiency-adjusted deferred acceptance (EADA) mechanism, suggested by \cite{kesten2010school}, runs DA and identified the last interrupter student to be rejected. The preferences of the interrupter student are then modified so that said school is no longer desired by the interrupter, and the DA is executed on the modified problem. The procedure is repeated until there are no more interrupters. EADA generates an allocation that is Pareto optimal and with a weakly smaller sum of ranks than that of DA, but that need not be rank-efficient or without justified envy. 

\section{Results}
\label{sec:results}
Our theoretical results relate to the properties of the expected allocation generated by RM, TTC and DA when students' preferences and schools' priorities are drawn independently and uniformly at random. This assumption is commonly used to analyze matching markets.\footnote{See \cite{che2018payoff} and references therein.} We study the asymptotic behavior of: i) expected average rank (efficiency), ii) expected maximum rank (inequality), and iii) expected number of students with justified envy generated by RM, TTC and DA in the next subsections.

\paragraph{Efficiency} We first study the expected average rank generated by RM, TTC and DA in random markets. To do so, we define $\overline{x} \coloneqq \frac{1}{n} \sum_{i=1}^n \rk_i(x_i)$, which denotes the average rank of the school to which students are assigned in allocation $x$.  

Proposition \ref{efficiency} shows that the expected average ranking in RM is smaller (i.e. better) than that in TTC and DA. It follows directly from a result by \cite{parviainen_2004} that has not yet been cited in the economics literature. In contrast, the results for DA and TTC are well-known and we simply restate them for completeness.

\begin{proposition}
	\label{efficiency}
	The expected average rank in RM, TTC and DA is:
	
	\begin{equation}
	\label{eq:avgrm}
	\lim _{n \to \infty} \EE [\overline x^\rmm] \leq 2
	\end{equation}
	
	\begin{equation}
	\label{eq:avgttc}
	\lim _{n \to \infty} \frac{\EE [\overline x^\ttc]}{\log n} = 1
	\end{equation}
	
	\begin{equation}
	\label{eq:avgda}
	\lim _{n \to \infty} \frac{\EE  [\overline x^\da ]}{\log n}=1
	\end{equation}
	
\end{proposition}

\begin{proof}
	Statement \ref{eq:avgrm} is proven by \citet[][Theorem 2.1, p. 105]{parviainen_2004}.\footnote{\cite{parviainen_2004} also provides a lower bound, and thus the expected average rank in RM is such that $\pi^2/6 \approx 1.65\leq \lim _{n \to \infty}  \EE [\overline x^\rmm] \leq 2$. }		
	Statement \ref{eq:avgttc} is proven by \citet[][equation 4, p. 439]{knuth1996}.\footnote{Knuth shows that $\EE[\sum_{i=1}^n \rk_i(x_i^\ttc)]=(n+1)H_n-n$, where $H_n$ is the $n$-th harmonic number and, therefore, $\lim_{n \to \infty} \frac{1}{n} \EE[\sum \rk(x_i^\ttc)] = \log (n)$. See also the note after the acknowledgements in \cite{frieze1995probabilistic}, p. 807. Even though Knuth focuses on matching with endowments, in which each agent owns one object, note that the preferences of each agent over his endowment are random, i.e. the object that an agents owns need not be his first-ranked object. The equivalence between matching with random ownership and without ownership has also been explored by \cite{abdulkadiroglu1998}.} 
	Statement \ref{eq:avgda} is proven by \citet[][Theorem 2, p. 538]{pittel1989}.
	
\end{proof}

Proposition \ref{efficiency} shows that the rank inefficiency of DA and TTC does not vanish as the market grows large because, even if the average rank obtained by DA and TTC grows slowly with the size of the market, the average rank obtained by RM is constant and does not grow with $n$. 

\paragraph{Inequality}  We measure inequality as the rank of the object obtained by the worst-off agent in the market, i.e. the maximum rank in the rank distribution. This measure follows John Rawls' idea that the welfare of a society is that of its worst-off member.\footnote{Alternatively, one could define inequality as the difference in ranks between the worst- and best-off agent. Because the rank of the object obtained by the best-off agent is 1 in any Pareto optimal allocation \citep{abdulkadiroglu1998}, both measures are equivalent.} To do so, we define $\underline{x} \coloneqq \max_{i} \rk_i(x_i)$, which denotes the rank of the object obtained by the worst-off agent in allocation $x$. 

Proposition \ref{inequality} shows that RM generates a significantly more egalitarian allocation than DA and TTC. In particular, TTC generates an allocation so unequal that the worst-off student  is assigned to a highly undesirable school in the lower half of their preference list. Such rank is much higher than the corresponding value for RM ($\log_2(n)$) and DA ($\log^2(n))$. 

\begin{proposition} The expected maximum rank of RM, TTC and DA is:
	\label{inequality}

\begin{equation}
	\label{eq:maxrm}
\lim _{n \to \infty} \frac{\EE [\underline x^\rmm]}{\log_2 (n)} = 1
\end{equation}	

\begin{equation}
	\label{eq:maxttc}
\lim _{n \to \infty} \frac{\EE [\underline x^\ttc]}{n} > 0.5
\end{equation}

\begin{equation}
	\label{eq:maxda}
\lim _{n \to \infty} \frac{\EE [\underline x^\da ]}{\log^2 (n)}=1
\end{equation}

\end{proposition}

\begin{proof}
Statement \ref{eq:maxttc} was proven by \citet[][p. 440]{knuth1996}.\footnote{\cite{knuth1996} shows that, in a serial dictatorship, the expected rank of the last dictator is $n/2$ (this is easy to see, as the last dictator has only one object to choose, and the expected rank of such object is exactly in the half of his preference list). Thus, when taking the maximum over the expected rankings of each dictator, the maximum must be greater than $n/2$. It is well-known that RSD is equivalent to TTC with random endowments, and the result follows.} Statement \ref{eq:maxda} was proven by \cite{pittel1992likely}, theorem 6.1, p. 382 and note before references, p. 400.

To prove statement \ref{eq:maxrm} we use the asymptotic rank distribution in RM. The probability that a student is assigned to their $i$-th choice is asymptotically equal to $\frac{1}{2^i}$ (Theorem 3 in \cite{aldous2001zeta}; see also Theorem 1.3 in \cite{parviainen_2004}), so that a student is assigned to a school with rank 1 with probability 1/2, to a school with rank 2 with probability 1/4, to a school with rank 3 with probability 1/8, to a school with rank 4 with probability 1/16, and so on. Such distribution is almost identical to that observed in a different problem, namely the number of consecutive heads in $n$ independent coin tosses, in which 0 heads obtains with probability 1/2, 1 heads with probability 1/4, 2 heads with probability 1/8, 3 heads with probability 1/16 and so on (the distribution of consecutive heads is simply that of expected ranks, but shifted by +1). The expected maximum in the latter problem (i.e. the longest run of heads) is known to be $\frac{\log(n/2)}{\log(2)}$  \citep{schilling2012surprising}. Therefore, the maximum rank in RM is equal to  $\frac{\log(n)/2}{\log(2)}+1=\log_2(n)$.\footnote{\citet[][Theorem 2, p. 1436]{frieze2007probabilistic} prove a similar result: the maximum cost in the cost assignment problem when costs are uniformly distributed in $[0,n]$ is  in the order of $\log(n)$.} 
\end{proof}

Although we only provide a lower bound for the maximum rank in TTC (of $0.5 \, n$), simulations suggest that the maximum rank in TTC converges to $0.63 \, n$.

\paragraph{Justified Envy} We use $e^\rmm, e^\ttc$ and $e^\da$ to denote the fraction of students who experience justified envy in the allocation obtained in RM, TTC and DA, respectively. Proposition \ref{envy} shows that RM generates fewer cases of expected envy than TTC, which is interesting since TTC is envy minimal in the class of strategy-proof and Pareto optimal mechanisms \citep{aeri}.

\begin{proposition}
	\label{envy}
	 The expected fraction of students with justified envy in RM, TTC and DA is:

\begin{equation}
	\label{eq:envyrm}
\lim _{n \to \infty}\EE [ e^\rmm] = 0.33	
\end{equation}

\begin{equation}
	\label{eq:envyttc}
	\lim _{n \to \infty} \EE [ e^\ttc] = 0.3863  
	\end{equation}
	
	\begin{equation}
		\label{eq:envyda}
\lim _{n \to \infty} \EE [ e^\da ]= 0
\end{equation}
\end{proposition}

\begin{proof}
	Statement \ref{eq:envyda} is well-known, as DA does not generate justified envy \citep{gale1962}. 
	
	For the remainder of the proof we use the fact that the number of students with justified envy in TTC ($e^\ttc$) and RSD ($e^\rsd$) is asymptotically equivalent \citep{che2017top}. Since schools' priorities are irrelevant in both RM and RSD, a student who is assigned to their $i$-th most preferred school does \emph{not} experience justified envy with probability $\frac{1}{2^{i-1}}$. To see this, notice that students placed into their 1st choice trivially do not experience justified envy with probability 1; students placed into their second best choice do not experience justified envy if the student who is accepted at his most preferred school has a higher priority than them, which occurs with probability 1/2; for students who are assigned to their third choice, they do not experience justified envy if their first and second most preferred school rank their assigned student above them, i.e. with probability $1/4$, and so on. 
	
	Thus, to obtain the total fraction of students who do not experience justified envy in RM and RSD (TTC), we just need to multiply i) the probability that a student matched to their $i$-th most preferred school does not experience justified envy, times ii) the fraction of students who are assigned to such a choice in RSD and RM. The fraction of students assigned to their $i$-th choice in RM asymptotically equals $\frac{1}{2^i}$ (Theorem 1.3 in \cite{parviainen_2004}), whereas in RSD the probability that the $k$-th dictator is assigned to his $j$-th most preferred school is given by $p_{k,j}=\frac{1}{n!}\binom{k-1}{j-1}(j-1)!(n-j)!(n+1-k)$.\footnote{See the derivation of this expression in the Appendix.} Putting these expressions together, and after some algebra for the RSD case detailed in the Appendix, we obtain:
	
\begin{equation}
e^\rmm=1- \sum_{i=1}^n \frac{1}{2^i} \times \frac{1}{2^{i-1}}=1 -  \sum_{i=1}^n\frac{1}{2^{2i-1}}\rightarrow 0.33 
\end{equation}

\begin{equation}
e^\rsd = 1- \frac{1}{n} \sum_{i=1}^n \sum_{j=1}^i\frac{1}{n!}\binom{k-1}{j-1}(j-1)!(n-j)!(n+1-k) \frac{1}{2^{i-1}} \rightarrow 0.3863 
\end{equation}	
	
	which finalizes the proof, since $\lim _{n \to \infty} \EE [ e^\ttc]  = \lim _{n \to \infty} \EE [e^\rsd].$
\end{proof}

\section{Data}
\label{sec:data}
One critique that can be made to our random market results is that they assume that students' preferences are independent, whereas students' preferences tend to be correlated, and such correlation may improve the performance of DA and TTC with regards to efficiency and equality. We show that this is not the case by using real-life data from secondary school admissions in Hungary in 2015. In summary, we find that TTC and DA perform even worse than when we assumed independent uniform preferences.

Our data contains the preferences and priorities of 10,131 students and 244 schools in Budapest. Because Hungary assigns students to schools using DA \citep{biro2008student}, we consider the reported preferences as truthful and apply DA, TTC, and RM to the reported preferences and priorities.\footnote{Because reported preferences lists are short (students rank only 4 schools on average), we also conduct a counterfactual analysis using the preferences estimated by \cite{aue2020happens}. We find an even larger efficiency and equality gap between RM and DA and TTC. We discuss these findings in the Appendix.}
When a student only ranks $k$ schools, we use $k+1$ as the rank of being unassigned. RM chooses the rank minimizing assignment randomly among all rank-efficient allocations. Figure \ref{fig:budapestreported} present the rank distribution realized under RM, DA, EADA and TTC using one RM allocation taken at random. Table \ref{tab:table1budapest} presents summary statistics averaging over 30 rank-minimal allocations.\footnote{There is minimal variance on the maximum rank across RM allocations (recall the theoretical result on footnote 9). We take the average over 30 rank-minimal allocations only because finding all rank-minimal allocations is computationally intensive.}
\begin{table}[h!]
	\centering
	\caption{Rank descriptive statistics for Budapest. Standard deviation for RM in parenthesis.}
	\label{tab:table1budapest}
	\begin{tabular}{lcccc}
		\toprule
		Variable $\backslash$ Mechanism	  	   & RM 		& TTC    		& DA & EADA	    \\
		\midrule
		Mean      		& 1.48 & 1.9 & 2.1 & 2.0    \\
 & \tiny (0.00) &  & &\\
				Maximum      	 &5.72 & 14 & 13  & 13	\\
				 & \tiny (0.45) &  & &\\
				 
		Share of students   & 0.46 & 0.45 &0 & 0.26 \\
		w. justified envy  &   \tiny (0.01)	&  		& & \\
				Blocking pairs    & 8,563 & 7,177 & 0 & 3,393 \\
								 & \tiny (29.33) &  & &\\
		Unassigned students & 2,559 & 2,508 & 2,704	& 2,704	 \\
		 & \tiny (4.64) &  & &\\
		\bottomrule
	\end{tabular}

	\footnotesize Note: The mean is computed dividing by the number of assigned students.\hfill{ }
\end{table}

The lessons we learn from computing the rank distributions in Budapest are similar to those we learned from looking at random markets. Table \ref{tab:table1budapest} shows that RM performs better than TTC and the currently used DA with regards to efficiency and equality. RM generates a better average rank (1.5) than TTC (1.9), DA (2.1) and EADA (2.0). RM assigns the average student to a school in their 16 percentile of their preference lists, whereas the corresponding percentile for DA and TTC are 35 and 29, respectively.
\begin{figure}[h!]
	\centering
	\includegraphics[width=.7\textwidth]{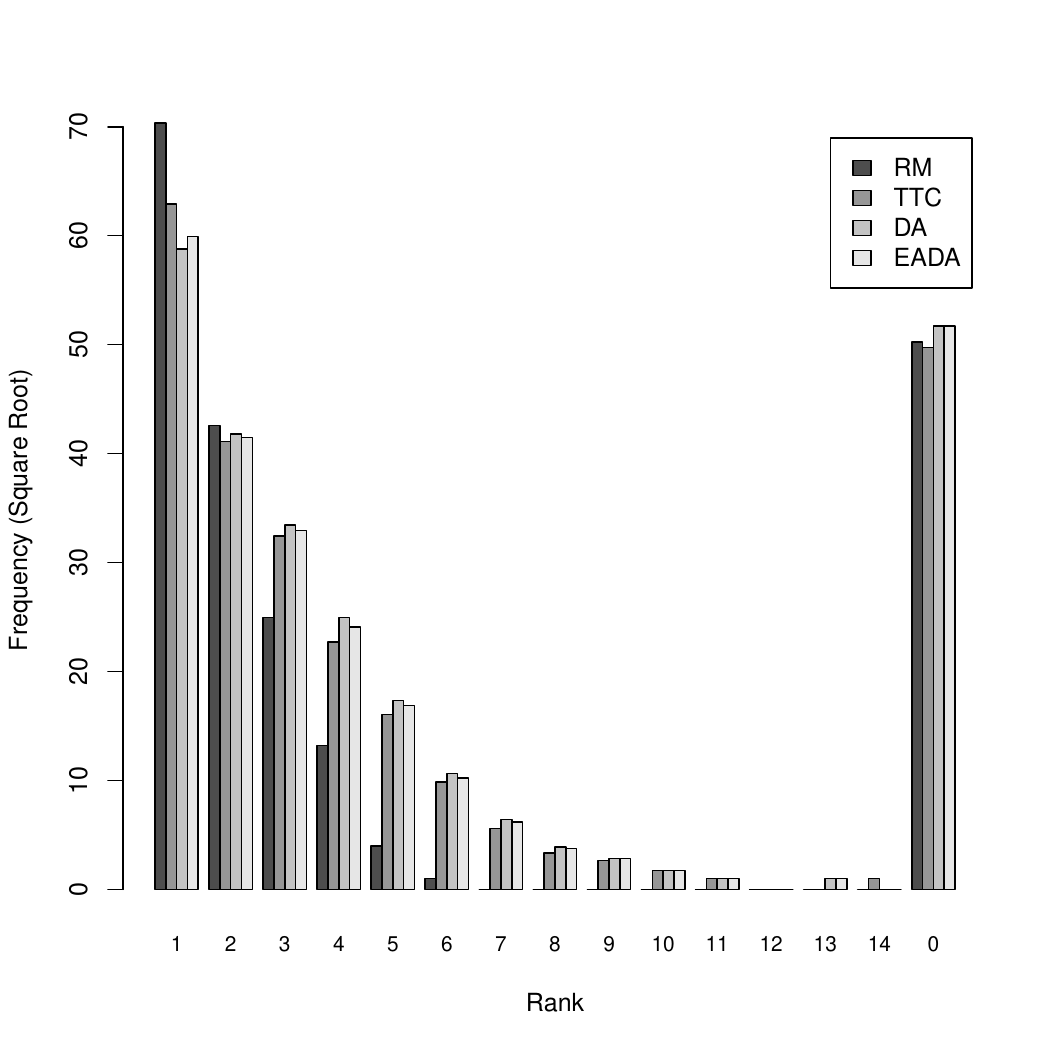} 
	\caption{Rank distribution generated for 10,131 students in the secondary school admissions in Budapest using reported preferences. The last bar (0) denotes unassigned students.}
	\label{fig:budapestreported}
\end{figure}

With regards to inequality, RM performs much better than DA, EADA and TTC, assigning the worst-off student to the 6th best choice rather than their 13th or 14th best. We find that DA and TTC are incomparable in terms of equality, since TTC assigns more students to a really undesirable school, but also assigns more students to a top school. The number of students unassigned in DA and EADA is higher than in RM (2,704 versus 2,558), which in turn is higher than in TTC (2,508)

We find that the rank distribution observed in EADA is only minimally better than one observed in DA (the average rank reduces from 2.1 to 2, whereas the maximum rank remains at 13, see Table \ref{tab:table1budapest} and Figure \ref{fig:budapestreported}). The rank distribution observed in EADA is significantly worse than the one produced by RM.
Our results suggests that EADA moderately improves the inefficiency of DA but fails to achieve the efficiency of RM. Moreover, EADA does not reduce the large inequality generated by DA.

Regarding justified envy, we observe that TTC and RM generate justified envy in roughly similar fraction of students (45\% and 46\%, respectively). One difference between TTC and RM is that, while they generate justified envy in the roughly the same number of students, RM may generate envy in students with higher priorities (e.g. students with good grades). We observe some evidence supporting this hypothesis. TTC and RM generate justified envy in almost the same number of students, but we observe 20\% more blocking pairs in RM than in TTC. EADA generates justified envy in significantly fewer students (26\%). 

The rank distributions are similar to those documented in other studies. \cite{che2018payoff} and \cite{aeri} also document that TTC assigns more students to their first choice than DA. Both studies also find that DA and TTC generate a similar number of unassigned students. 

Our empirical analysis uses the preferences that students submit in DA to generate the TTC and RM allocations. A concern is that students would submit different preferences when allocations are determined by RM, which is not strategy-proof. To mitigate this concern, we compute the rank distribution generated by RM when a fraction of the students who have incentives to misrepresent their preferences do so (from 20\% to 80\%). We find that the rank distribution and number of students with justified envy remain largely unchanged, even when a fraction of agents misrepresent their preferences.\footnote{The average rank varies minimally because the decrease in rank from the students who misrepresent their preferences is almost perfectly counteracted by the decrease in rank of the students who become worse off; both changes are usually small, around 1 or 2 ranks. See the Appendix for detailed summary statistics. A detailed analysis of the scope of manipulations in the rank-minimizing mechanism remains an interesting open question.}

In our view, it is unclear whether students would misrepresent their preferences in RM. The potential gains from manipulation are tiny (the average student can only improve by less than one rank in their preference list with iid preferences, and by less than 2 ranks in the data), and manipulations are risky and could lead to worse outcomes. Moreover, RM is not obviously manipulable, and thus cannot be manipulated by cognitively limited agents \citep{troyan2022non}. Furthermore,  there is evidence of high truth-telling rates in not obviously manipulable mechanisms \citep{cerrone2022school}.\footnote{\cite{cerrone2022school} find that almost twice as many people (70\% versus 40\%) behave truthfully in the efficiency adjusted deferred acceptance (EADA) mechanism versus standard DA, even though DA is strategy-proof and EADA is not.} Conducting an experiment comparing RM, DA and TTC would clarify how strong is our assumption of truthful behavior in RM. We leave this interesting question for future research.

\section{Conclusion and Open Questions}
Our paper does not argue that strategy-proofness should be abandoned as key desiderata in school choice. Strategy-proofness is a desirable property and has a clear justification in terms of levelling the playing field across sophisticated and unsophisticated applicants. Thus, it is not unreasonable for policymakers to use strategy-proof mechanisms such as deferred acceptance and top trading cycles. However, our paper makes the case that academics and policymakers should be aware that strategy-proofness involves significant costs in terms of efficiency and equality.

Our paper also highlights the remarkable properties of the rank-minimizing mechanism, which has received little attention in the literature.  Its outstanding efficiency and equality properties in theory and practice are strong arguments for its use in some situations, particularly when schools' priorities are random lotteries, as in Brighton and Hove \citep{allen2013short} and Amsterdam \citep{oosterbeek2020using}.\footnote{The UK official School Admissions Code 2007 and the report by \cite{coldron2008secondary} propose the use of lotteries to mitigate segregation; see also \cite{basteck2021lotteries}. In Northern Ireland, secondary schools' priorities are not allowed by law to depend on academic selection \citep{brown2021rise}. In England, a small but growing number of schools use lotteries as the main admissions criterion \citep{noden2014banding}.}  We also conjecture that RM may generate less segregated allocations than those generated by DA or TTC because it does not use schools' priorities.\footnote{There is evidence that school choice increases segregation by ethnic and family background, e.g. \cite{soderstrom2010school}.} On the other hand, potential constraints to implementing the rank-minimizing mechanism include its non strategy-proofness, its lack of stability and a lack of transparent description of the mechanism, which could be a problem in places where the trust in the corresponding Education Authority is low. 

In this paper, we have shown the efficiency and equality costs that arise when using two prominent strategy-proof mechanisms. An interesting open question is: what is the minimum loss that can arise in \emph{any} strategy-proof mechanism? In other words, what is the smallest expected average and/or maximum rank generated by a strategy-proof mechanism in iid random markets? We conjecture that no strategy-proof mechanism significantly improves on TTC, i.e. that no strategy-proof mechanism achieves an expected average rank asymptotically smaller than  $\log(n)$ or an expected maximum rank below $0.5 n$. Any such mechanism, if exists, must lack Pareto optimality.\footnote{A growing literature treats the stability property as binding, and asks for more efficient mechanisms that satisfy it by itself \citep{abdulkadiroglu2021school} or jointly with ordinal dominance incentive compatibility \citep{bodoh2020optimizing}.}

\section*{Note after Acceptance}
After this work was completed, \cite{ortegaziegler} have computed the expected rank generated by EADA in random, iid one-to-one markets, which is in the order of $\log(n!)/n$. Their analysis builds on an alternative implementation of EADA proposed by \cite{tang2014new}.  Their finding, combined with our Theorem 1, shows that the gap between the average ranks in EADA and RM grows as $n$ increases.

\section*{Acknowledgements}
We thank the anonymous referees of this journal and the MATCH-UP 2022 workshop for their valuable suggestions. We also acknowledge helpful comments from Mustafa Afacan, David Delacrétaz, Aytek Erdil, Takashi Hayashi, Yoan Hermstrüwer, Jörgen Kratz, Aditya Kuvalekar, David Manlove, Vincent Meisner, Antonio Miralles, Hervé Moulin, Afshin Nikzad, Juan Sebasti\'an Pereyra, Erel Segal-Halevi, Olivier Tercieux, Peter Troyan, Bertan Turhan, Utku Ünver, Mark Wilson, Bumin Yenmez, Gabriel Ziegler, and audiences at the 2023 Econometric Society European Summer Meeting, Match-UP, the Conference on Economic Design, the Coalition Theory Network Workshop, the Belfast Easter Workshop on School Choice and seminars at ISER Essex, Boston College, the University of Glasgow and NTNU.

We are indebted to Sarah Fox, Taylor Knipe and Abbas Ali Shah for proofreading this paper.
This paper previously circulated under the title ``Improving Efficiency and Equality in School Choice''.
We acknowledge funding from the ESRC, the British Academy and the Leibniz Association, as part of project K125/2018: ``Improving school admissions for diversity and better learning outcomes".

\newpage
\section*{Appendix A - Proof of Proposition 3}
\label{sec:proof}

In RSD the probability that the $k$-th dictator is assigned to his $j$-th most preferred school is given by 
\begin{equation}
\label{eq:ap1}
p_{k,j}=\frac{1}{n!}\binom{k-1}{j-1}(j-1)!(n-j)!(n+1-k)
\end{equation}

Where $\binom{k-1}{j-1}(j-1)!$ denotes the possible combinations in which exactly the $j-1$ most preferred schools by dictator $k$ have been chosen by the previous $k-1$ dictators; $(n-j)!$ denotes the arbitrary combinations in which the schools ranked worse than $j$ can appear in the preferences of the $k$-th dictator, and the term $n-(k-1)$ denotes the possibilities that the $j$ most preferred school for dictator $k$ has not been chosen by the previous $k-1$ dictators. The normalization by $n!$ is to account for the number of all possible preferences profiles. For example, if $k=1$, then $p_{1,1}=1$ (the probability that the first dictator gets his first school is one) and  $p_{1,j}=0$ for any $j>1$ (the probability that the first dictator gets a school worse than his top one is zero). Similarly, when $j=1$, then  $p_{k,1}=\frac{n+1-k}{n}$ (this is the probability that the $k$-th dictator gets his top school, or equivalently, the probability that none of the $k-1$ dictators before him are assigned to the school that dictator $k$ ranks as first). Note that ${n \choose 0} = 1$ and ${n \choose m} =0$ for any $m>n$, so that the probability that dictator $k$ is assigned to a school with a rank higher than $k$ is zero. Using a simplification by \citeauthor{simplification}, equation \eqref{eq:ap1}	can be rewritten as:
\begin{eqnarray}
p_{k,j}&=& (n+1-k) \frac{(k-1)!(j-1)!(n-j)!}{n!(j-1)!(k-j)!} \\
&=& (n+1-k) \frac{(k-1)!(j-1)!(n-j)!}{n!(j-1)!(k-j)!} \frac{k}{k} \frac{j}{j}\\
&=& \frac{(n+1-k)}{k} \frac{\binom{k}{j}}{\binom{n}{j}}
\end{eqnarray}

Since RSD is independent of schools' priorities, a student placed in their $j$-th most preferred school does \emph{not} experience justified envy with probability $\frac{1}{2^{j-1}}$. Therefore, the total number of students without justified envy in RSD $(NE^\rsd)$ equals
\begin{eqnarray}
	NE^\rsd&=&\sum_{k=1}^{n}	\sum_{j=1}^{n} \frac{(n+1-k)}{k} \frac{\binom{k}{j}}{\binom{n}{j}}\frac{1}{2^{j-1}}\\
		&=&\sum_{j=1}^{n}	\frac{1}{\binom{n}{j}2^{j-1}} \sum_{k=j}^{n} \frac{(n+1-k)}{k} \binom{k}{j}\\
				&=&\sum_{j=1}^{n}\frac{1}{\binom{n}{j}2^{j-1}} \left[	 (n+1)\underbrace{\sum_{k=j}^{n} \frac{1}{k} \binom{k}{j}}_{A} -  \sum_{k=j}^{n}  \binom{k}{j} \right]
\end{eqnarray}

Where
\begin{eqnarray}
	A&=&\sum_{k=j}^{n} \frac{1}{k} \binom{k}{j}=\sum_{k=j}^{n} \frac{1}{k} \frac{k!}{j!(k-j)!}=\sum_{k=j}^{n}  \frac{(k-1)!}{j\, (j-1)!(k-j)!}\\
&=&\frac{1}{j} \sum_{k=j}^{n}   \binom{k-1}{j-1}=\frac{1}{j} \sum_{k=j-1}^{n-1}   \binom{k}{j-1}
\end{eqnarray}

Plugging this in our expression for $NE^\rsd$, we have
\begin{eqnarray}
NE^\rsd=\sum_{j=1}^{n}\frac{1}{\binom{n}{j}2^{j-1}} \left[	 \frac{n+1}{j}\sum_{k=j-1}^{n-1} \binom{k}{j-1} - \sum_{k=j}^{n}  \binom{k}{j} \right]
\end{eqnarray}

Using the Hockey-stick identity $\sum_{k=j}^{n} \binom{k}{j} =\binom{n+1}{j+1}$, we obtain

\begin{eqnarray}
	NE^\rsd&=&\sum_{j=1}^{n}\frac{1}{\binom{n}{j}2^{j-1}} \left[	 \frac{n+1}{j}\binom{n}{j} -\binom{n+1}{j+1} \right]\\
&=&\sum_{j=1}^{n}\frac{1}{\binom{n}{j}2^{j-1}} \left[	 \frac{n+1}{j}\binom{n}{j} -\frac{(n+1)\, n!}{(j+1) \, j!(n-j)!} \right]\\
&=&\sum_{j=1}^{n}\frac{1}{\binom{n}{j}2^{j-1}} \left[(n+1)	\left( \frac{1}{j}-\frac{1}{j+1}\right)\binom{n}{j} \right]\\
&=&(n+1)\sum_{j=1}^{n}\left(\frac{1}{2}\right)^{j-1} \left( \frac{1}{j}-\frac{1}{j+1} \right)
\end{eqnarray}

We divide both sides by $n+1$ to obtain
\begin{eqnarray}
	\frac{NE^\rsd}{n+1}&=&\sum_{j=1}^{n}\left(\frac{1}{2}\right)^{j-1} \frac{1}{j} -\sum_{j=1}^{n}\left(\frac{1}{2}\right)^{j-1} \frac{1}{j+1} \\
	&=&2\sum_{j=1}^{n}\left(\frac{1}{2}\right)^{j} \frac{1}{j} -4\sum_{j=1}^{n}\left(\frac{1}{2}\right)^{j+1} \frac{1}{j+1} \\
		&=&2\sum_{j=1}^{n}\left(\frac{1}{2}\right)^{j} \frac{1}{j} -4\sum_{j=2}^{n+1}\left(\frac{1}{2}\right)^{j} \frac{1}{j} \\
				&=&-2\sum_{j=2}^{n}\left(\frac{1}{2}\right)^{j} \frac{1}{j} +1 -\frac{4}{(n+1)}\left(\frac{1}{2}\right)^{n+1}\\
								&=&-2\sum_{j=2}^{n}\left(\frac{1}{2}\right)^{j} \frac{1}{j} +1 -\frac{4}{(n+1)}\left(\frac{1}{2}\right)^{n+1}-1+1\\
												&=&-2\sum_{j=1}^{n}\left(\frac{1}{2}\right)^{j} \frac{1}{j} +2 -\frac{4}{(n+1)}\left(\frac{1}{2}\right)^{n+1}\\
														&=&2-\underbrace{\frac{1}{(n+1)}\left(\frac{1}{2}\right)^{n-1}}_{B}-2\underbrace{\sum_{j=1}^{n}\left(\frac{1}{2}\right)^{j} \frac{1}{j} }_{C}
													\end{eqnarray}
As $n$ goes to infinity, $B$ goes to 0 and $C$ is the Taylor expansion for $-\ln(1/2)$.
\begin{eqnarray}
\lim_{n \to \infty} 	\frac{NE^\rsd}{n}	=\lim_{n \to \infty} 	\frac{NE^\rsd}{n+1}=  2+2\ln(\frac{1}{2})=0.6137
\end{eqnarray}

Thus, the fraction of students who experience justified envy in RSD tends to $e^\rsd = 1-0.6137=0.3863$, which is what we wanted to prove. 

\section*{Appendix B - Simulations}
\label{sec:simulations}
In simulated markets (see Tables \ref{tab:table1} and \ref{tab:table2}), we clearly see that RM dominates TTC and DA in efficiency (average rank) and inequality (maximum rank). Given the large ranks that realize in TTC, it is unsurprising that the rank distribution is large too. The variance of RM is much smaller, which shows that the ranks are heavily concentrated among the first four top choices. Table \ref{tab:table1} also allows us to assess the accuracy of the random market results presented in section \ref{sec:results}. For TTC, the mean rank is surprisingly close to the theoretical prediction ($\pm 1$ of $\log (n)$). In RM, the upper bound provided of 2 for the mean is quite tight, and the approximation $\log_2(n)$ for the max rank is also remarkably accurate.
\begin{table}[h!]
	\centering
	\caption{Rank descriptive statistics. Average over 1,000 simulations.}
	\label{tab:table1}
	\begin{tabular}{lcccc}
		\toprule
		\multirow{1}{*}{	Variable $\backslash$ Mechanism	 } &  RM & TTC   &DA & EADA\\
\midrule
		Mean      			   & 1.8 & 4.3 &5.0 & 	2.7\\
		Max      			     & 6 & 64   	& 23 &	16\\
		Variance   	 		  & 1 & 79  &20	& 7 \\
Blocking pairs				& 41 & 150  & 0		    &  19  \\
		Justified envy				& 30 & 37  & 0		    &  13\\
		\bottomrule
	\end{tabular}
\end{table}

The severity of the inequality generated by TTC is fully exposed in Table \ref{tab:table2}. TTC not only makes someone really worse off, assigning them a really bad object ($0.63 \, n$), but it assigns an object in the bottom 90\% (not top 10\%) of their preferences to over 1.5\% of the agents. In contrast, RM does not assign such a poor option to any agent. RM also assigns more agents to their top choice than TTC. 
\begin{table}[h!]
	\centering
	\caption{Percentage of agents who receive an object with rank higher (worse) than $m$. Average over 1,000 simulations.}
	\label{tab:table2}

		\begin{tabular}{lcccc}
			\toprule
		\multirow{1}{*}{	$m$ $\backslash$ Mechanism	 } &  RM & TTC   &DA & EADA\\
\midrule
			$1$          & 47 & 50 & 79 & 60        \\
			$2$         & 20 & 33 & 63 & 37       \\
			$\log(n)$    & 3  & 19 & 40 & 15      \\
			$0.1\, n$  & 0  & 8   &10& 2     \\
			$0.25\, n$  & 0 & 3    & 0& 0    \\
			$0.5\, n$   & 0 & 1   &0 & 0    \\
			\bottomrule
		\end{tabular}
	
\end{table}

\section*{Appendix C - Strategic applicants}
\label{sec:strategies}

To understand the impact of manipulating students on the rank distribution generated by RM, we compute the RM allocation when a fraction (20\%, 40\%, 60\% and 80\%) of students who have incentives to manipulate do so. We consider two possible manipulations.

\textit{Drop-Assigned}: In the reported preference data from Budapest, 5,183 students are not assigned their first preference in the RM.
A first possible manipulation we consider is for them to move their assigned school to the end of their preference list in the hope of increasing admissions chances at a more preferred school. That is, a student who ranks $s_1 \succ \cdots \succ s_i \succ \cdots \succ s_n$ and gets $s_i$ will perform a manipulation of the form $s_1 \succ \cdots \succ s_n \succ s_i$. This is equivalent to not ranking the school to which the applicant would have been assigned.

\textit{Drop-First}: Another 3,369 students are neither assigned their first nor their second preference in the RM.
A second possible manipulation we consider is for them to move their first preference to the end of their preference list in the hope of increasing admissions chances at their second preference. That is, a student who ranks $s_1 \succ s_2 \succ s_3 \succ \cdots \succ s_n$ and gets $s_3$ or worse will perform a manipulation of the form $s_2 \succ s_3 \succ \cdots \succ s_n \succ s_1$. This manipulation has been observed in real life applications \citep{abdulkadiroglu1998}.

Table \ref{tab:manipulations3} shows that the summary statistics in RM remain largely unchanged in the presence of strategic applicants. The statistics for the RM with a share of 0\% strategic applicants are equivalent to the RM results with reported preferences in Table \ref{tab:table1budapest}. With an increasing share of strategic applicants, the average and maximum rank change marginally, remaining well below the corresponding ranks for DA and TTC. Justified envy remains about the same. The number of unassigned students slightly decreases for the first manipulation, and slightly increases for the second one.

\begin{table}[h!]
	\centering
	\caption{Rank descriptive statistics for RM with strategic applicants.}
	\label{tab:manipulations3}
	\begin{tabular}{lccccc}
		\toprule
		Reported Preferences	  	   &  \multicolumn{5}{c}{Share (number) of strategic applicants} \\
		Variable $\backslash$ RM mechanism 	   & 0\% & 20\% & 40\% & 60\% & 80\% \\
		 									   & (0) 	& (1,036)   & (2,073) & (3,109) & (4,146) \\		
		\midrule
		\multicolumn{6}{l}{\textit{Panel A: Drop-Assigned}}\\
		\midrule
		Mean      			&  1.48  &  1.50  &  1.51  &  1.53  &  1.54   \\
		Maximum      		&  6    	&  7			& 6 &6 & 7    	\\
		Variance   	 		&  0.58 &   0.70  &  0.70  &  0.81  &  0.82 	\\
		Share of students   &  0.46  &  0.45 &   0.45  &  0.46  &  0.46 \\
		w. justified envy  &   	&  		& && \\
		Unassigned students  & 2,555 & 2,599 & 2,590 & 2,636 & 2,625 		 \\
		\midrule
\multicolumn{6}{l}{\textit{Panel B: Drop-First}}\\
\midrule
		Mean      			&  1.48  &  1.50  &  1.52  &  1.54 &   1.56   \\
		Maximum      		&  6    	&  6			& 6 &6 & 6    	\\
		Variance   	 		&  0.58 &   0.61  &  0.63   & 0.65  &  0.67 	\\
		Share of students   &  0.46  &  0.46 &   0.46  &  0.47  &  0.47 \\
		w. justified envy  &   	&  		& && \\
		Unassigned students  & 2,555 & 2,526 & 2,503 & 2,479 & 2,449 		 \\
		\bottomrule
	\end{tabular} 
\end{table}

Figure \ref{fig:manipulations3} presents the rank distributions. Overall, the results show that the rank distributions remain largely unchanged, even for large shares of strategic applicants.

\begin{figure}[h!]
	\begin{subfigure}{.55\textwidth}
		\includegraphics[width=.65\textwidth]{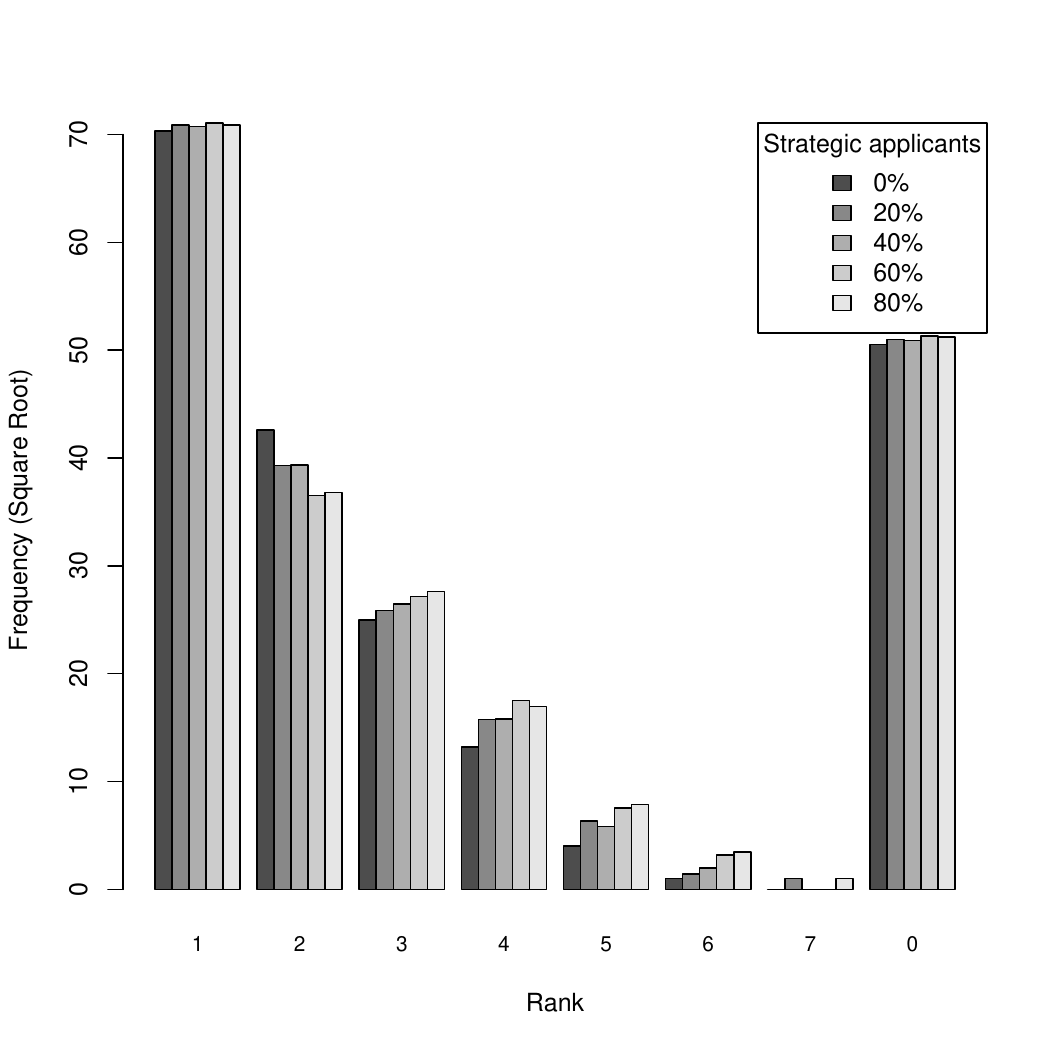}
		\caption{\textit{Drop-Assigned}.}
	\end{subfigure}
\begin{subfigure}{.55\textwidth}
	\includegraphics[width=.65\textwidth]{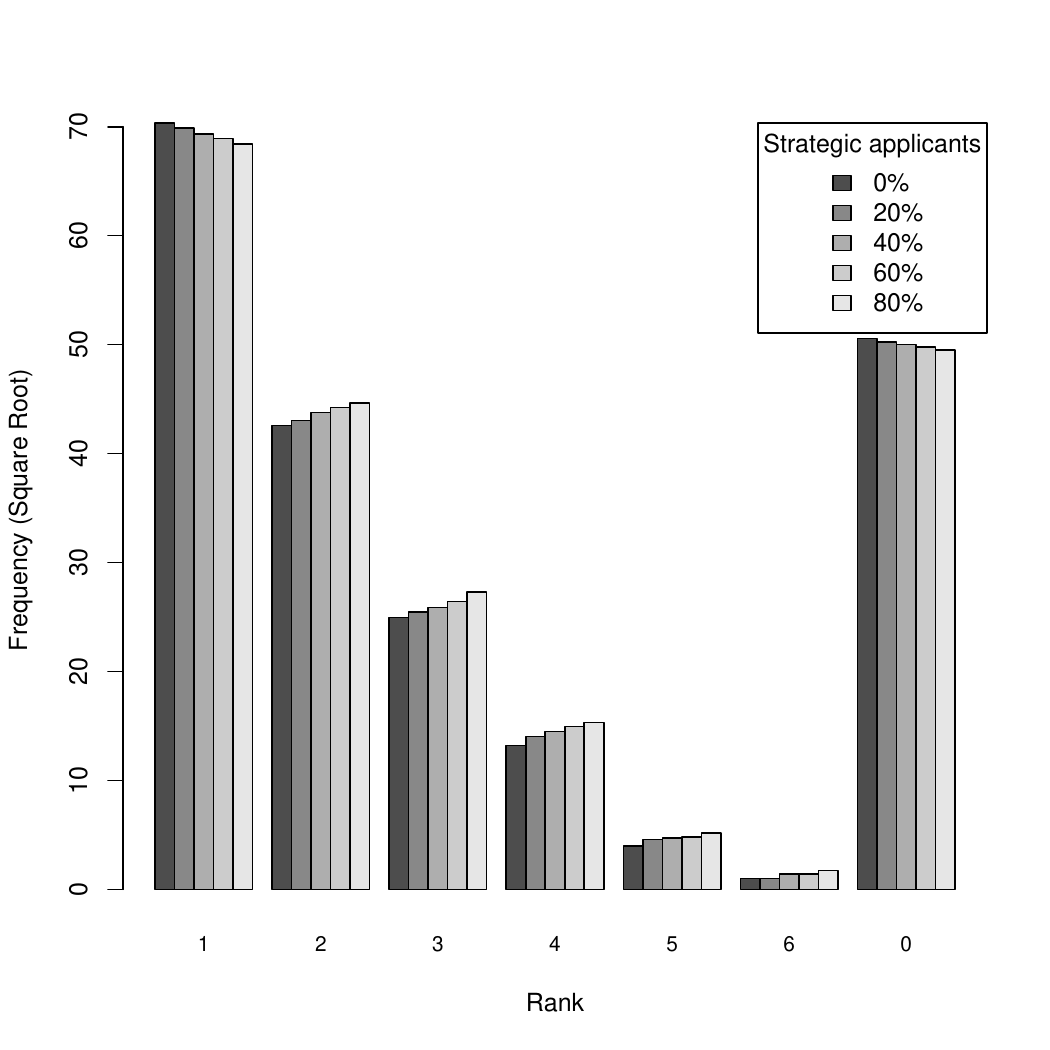}
	\caption{\textit{Drop-First}.}
	\end{subfigure}
	\caption{Rank distribution generated by RM for 10,131 students in the secondary school admissions in Budapest using reported preferences by fraction of strategic applicants. The last bar (0) denotes unassigned students.}
		\label{fig:manipulations3}
\end{figure}

\section*{Appendix D - One-to-One versus Many-to-One}
\label{sec:many}

In the main text, we imposed the assumption that each school has one seat. Here, we relax this assumption and allow each school to have $k$ seats. We find that this change does not affect the conclusions presented in the main text.

Here we assume that each market includes $n$ students, $n/k$ schools and a total of $n$ seats. The priorities of schools over sets of students are responsive, so that a school comparing two assignments that differ in only one student, prefers the assignment containing the more preferred student \citep{roth1985college}. 

In Figure \ref{fig:many} we present the rank distribution generated with $n=100$ students, $k=10$ seats in each school and 10 schools. We observe that the distributions mimic those presented in Figure \ref{fig:histogram}.
\begin{figure}[h!]
	\centering
	\includegraphics[width=0.8\textwidth]{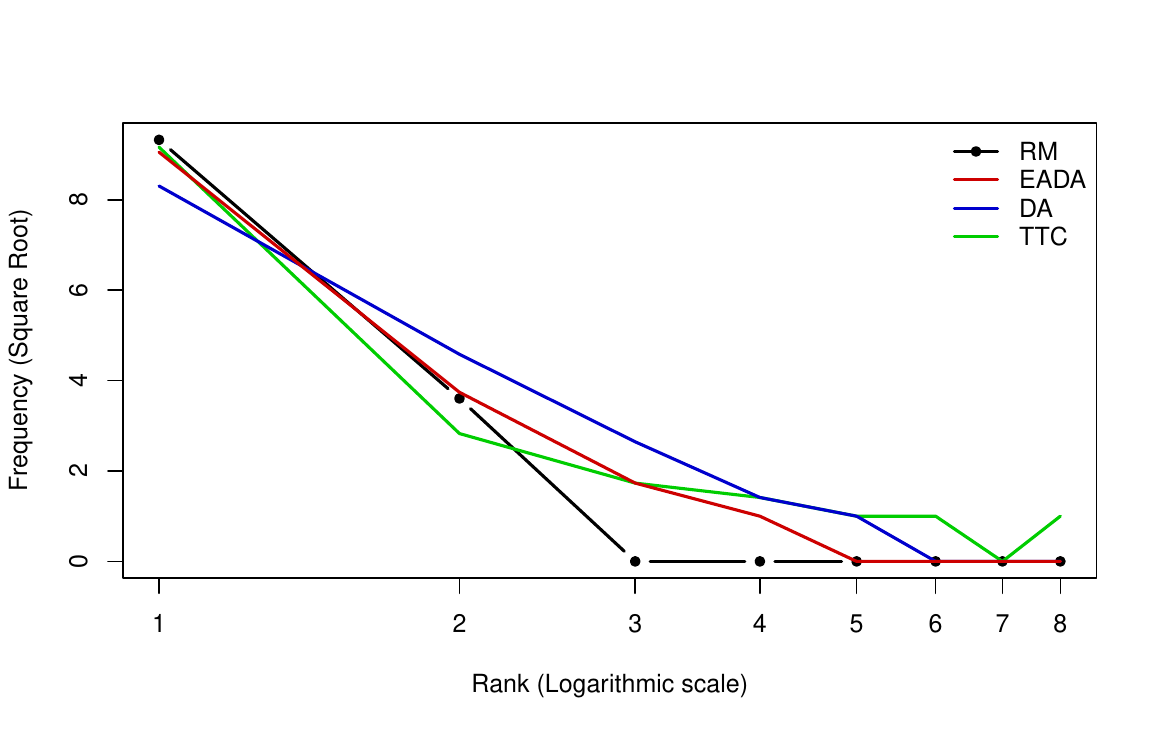} 
	\caption{Rounded average rank distribution in 1,000 iid many-to-one random markets with $n=100$ and $k=10$. \tiny The rank vector obtained in each market is sorted in descending order, and then we compute the average  across all markets coordinate-wise (results are rounded to the nearest integer). The x-axis is truncated at the highest value with positive density. EADA refers to Kesten's efficiency adjusted DA.}
	\label{fig:many}
\end{figure}

In summary, we still observe that the rank distribution generated by RM clearly dominates the ones obtained with DA, EADA and TTC. TTC still generates the most unequal distribution with the largest maximum rank. EADA significantly improves the efficiency of DA, and modestly improves its equality. These simulations suggest that the theoretical results we obtained for one-to-one markets carry over to this many-to-one scenario, with the corresponding values divided by $k$ (except the average rank in RM, which is likely to converge to 1 in this set-up).

\section*{Appendix E - Estimated Preferences}
\label{sec:estimated}

We analyse the allocation generated by DA, EADA, TTC and RM when estimated preferences are used, as in \cite{aue2020happens}. These preferences are complete, meaning that students rank all schools, unlike in our analysis in the main text where students tend to rank just a few (average 4.4), and provide insights into the performance of our four algorithms with longer preference lists. Figure \ref{fig:budapestcomplete} summarizes our findings. 

\begin{figure}[h!]
	\centering
	\includegraphics[width=\textwidth]{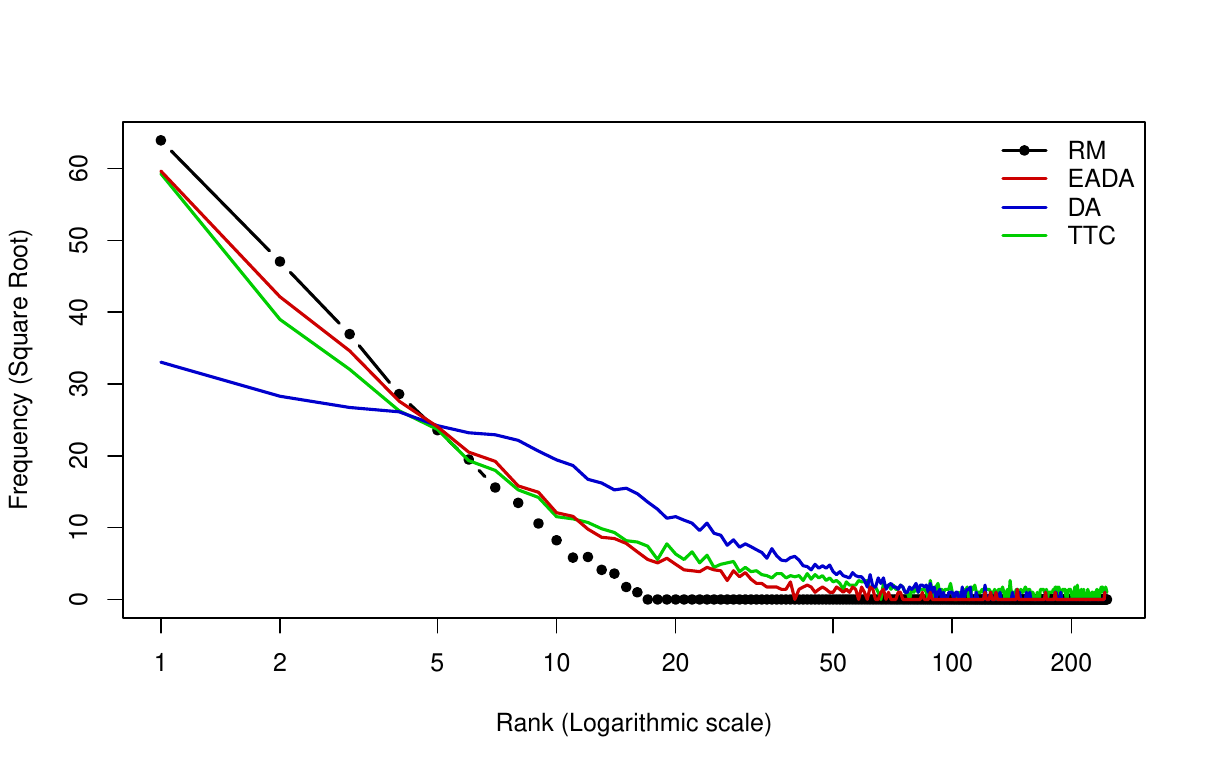} 
	\caption{Rank distribution generated in the secondary school admissions in Budapest. \tiny The number of students and schools is 10,131 and 244, respectively. Thus, the maximum rank is 244. See Section \ref{sec:data} for details.}
	\label{fig:budapestcomplete}
\end{figure}

The average student substantially improves their placement in RM (average rank in RM is 2.7, compared to 8.9 in TTC and 12.3 in DA). With regards to inequality, RM performs much better than DA and TTC with complete preferences, assigning the worst-off student to the 16th best choice rather than to their 241th and 244th, respectively (out of 244). It also assigns less than 2\% of the student population to their 10th ranked school or worse, whereas TTC and DA assign 16\% and 41\% of the student population to such school, respectively.  With estimated preferences, the efficiency improvement of EADA over DA becomes more evident: the average rank decreases from 12.3 to 4.4, although the maximum rank remains unchanged at 241 (out of 244). Either with reported or estimated preferences, the rank distribution observed in EADA is significantly worse than the one produced by RM. With estimated preferences, EADA generates justified envy in almost the same fraction of students as RM and TTC (EADA 56\%, RM 58\%, TTC 64\%), and even more blocking pairs than RM (20,529 in EADA versus 15,441 in RM). 

\end{document}